\documentclass[<options>]{elsarticle}
\biboptions{longnamesfirst,angle,semicolon}

\usepackage[utf8]{inputenc}   
\usepackage[english]{babel}   
\usepackage{enumerate}   
\usepackage{verbatim}   
\usepackage[margin=2.5cm]{geometry}   
\usepackage{amssymb}   
\usepackage{amsthm}   
\usepackage{amsmath}   
\usepackage{babel}   
\usepackage{csquotes}   
\usepackage{graphicx}   
\usepackage{caption}   
\usepackage{listings}   
\usepackage{bbm}   
\usepackage{todonotes}   
\usepackage{color,soul}   
\usepackage{tabularx,colortbl}   
\usepackage{tcolorbox}   
\usepackage{mathrsfs}
\usepackage{mathtools}
\usepackage[makeroom]{cancel}
\usepackage{tikz}
\usepackage{subcaption}
\usepackage[labelformat=parens,labelsep=quad,skip=3pt]{caption}
\usepackage{graphicx}

\theoremstyle{plain}

\input amssym.def   
\input amssym.tex   
 
\def\beq{\begin{equation}}   
\def\eeq{\end{equation}}   
\def\beqn{\begin{eqnarray}}

\def\eeqn{\end{eqnarray}}

\def\utilde#1{\mathord{\vtop{\ialign{##\crcr   
$\hfil\displaystyle{#1}\hfil$\crcr\noalign{\kern1.5pt\nointerlineskip}   
$\hfil {}\hfil$\crcr\noalign{\kern1.5pt}}}}}   
 
\theoremstyle{definition}   
\newtheorem{theorem}{Theorem}[section]

\newtheorem*{proposition*}{Proposition}   
   
\theoremstyle{remark}   
\newtheorem*{remark}{Remark}

\newtheorem*{example*}{\bf Example}   
\newtheorem{example}{\bf Example}

\theoremstyle{definition}   
\newtheorem{definition}{Definition}

\def\var{\varepsilon}

\def\R{\mathbb R}

\def\to{\rightarrow}   
\def\dlim[#1][#2]{\lim_{#1 \to #2, #1 \neq #2}}

            \usepackage[utf8]{inputenc}   
\usepackage[english]{babel}   
\usepackage{enumerate}   
\usepackage{verbatim}   
\usepackage{amssymb}   
\usepackage{amsthm}   
\usepackage{amsmath}   
\usepackage{babel}   
\usepackage{csquotes}   
\usepackage{graphicx}   
\usepackage{caption}   
\usepackage{listings}   
\usepackage{bbm}   
\usepackage{todonotes}   
\usepackage{color,soul}   
\usepackage{tabularx,colortbl}   
\usepackage{tcolorbox}   
\usepackage{mathrsfs}
\usepackage{mathtools}
\usepackage[makeroom]{cancel}
\usepackage{tikz}
\usepackage{extarrows}
\usepackage{multicol}
\usepackage{multirow}
\usepackage{subcaption}
\usepackage[labelformat=parens,labelsep=quad,skip=3pt]{caption}
\usepackage{graphicx}

\theoremstyle{plain}

\input amssym.def   
\input amssym.tex

\def\beq{\begin{equation}}   
\def\eeq{\end{equation}}   
\def\beqn{\begin{eqnarray}}

\def\eeqn{\end{eqnarray}}

\def\utilde#1{\mathord{\vtop{\ialign{##\crcr   
$\hfil\displaystyle{#1}\hfil$\crcr\noalign{\kern1.5pt\nointerlineskip}   
$\hfil {}\hfil$\crcr\noalign{\kern1.5pt}}}}}

\def\var{\text{Var}}

\def\L{{\cal L}}

\def\R{\mathbb R}

\def\to{\rightarrow}   
\def\dlim[#1][#2]{\lim_{#1 \to #2, #1 \neq #2}}

\DeclareMathOperator*{\argmax}{arg\,max}             
\makeatletter
\def\ps@pprintTitle{%
  \let\@oddhead\@empty
  \let\@evenhead\@empty
  \let\@oddfoot\@empty
  \let\@evenfoot\@oddfoot
}
\makeatother

\begin{document}

\begin{frontmatter}

\title{Bayesian Information Criterion for Linear Mixed-effects Models}

\author[1]{Nan Shen\corref{cor1}%
}
\ead{nli2@niu.edu}
\author[1]{Bárbara González}
\ead{bgonzalez4@niu.edu}

\cortext[cor1]{Corresponding author}

\address[1]{Northern Illinois University, Department of Statistics and Actuarial Science, 1425 Lincoln Hwy, DeKalb, IL 60115}

\begin{abstract}
The use of Bayesian information criterion (BIC) in the model selection procedure is under the assumption that the observations are independent and identically distributed (i.i.d.). However, in practice, we do not always have i.i.d. samples. For example, clustered observations tend to be more similar within the same group, and longitudinal data is collected by measuring the same subject repeatedly.  In these scenarios, the assumption in BIC is not satisfied. The concept of effective sample size is brought up and improved BIC is defined by replacing the sample size in the original BIC expression with the effective sample size, which will give us a better theoretical foundation in the circumstance that mixed-effects models involve. Numerical experiment results are also given by comparing the performance of our new BIC with other widely used BICs. 

\end{abstract}


\end{frontmatter}

\section{Introduction}

\subsection{Understanding BIC and Bayesian model selection}
BIC is one of the approximation methods of Bayes factor as brought out by Schwarz in 1978 \cite{chiwazi}. Due to the computational difficulties, usually, exact Bayes factor is not used directly in model selection. Alternatively,  Bayes factor is always approximated by BIC or a variant BIC \cite{ray} using Laplace method for approximation integrals. The deduction of the expression of BIC is under the assumption that the observations $y_1, y_2, \cdots, y_n$ are i.i.d. in which the Hessian matrix becomes the observed Fisher information matrix. However, when there is non-independence in the data, we cannot use BIC directly. Our research generalized the use of BIC in the circumstances that the linear mixed-effects models involve.

\subsubsection{Bayes Factor}
Bayes factor plays a core role in Bayesian model comparison. It determines how far one collection of information should alter one's level of belief in one model versus another. In using Bayes factors, it is essential to calculate the marginal likelihood of two models.

Consider the circumstances that we are doing the model comparison. Suppose we have a pool of models that we would use to describe a given phenomenon. In other words, we want to find out which of them fits the data best. Thinking probabilistically about this. One way to frame the problem would be to calculate the posterior probability that we ascribe to model 1 ($M_1$) conditional to the data that we collect, and we can compare that with the posterior probability of model 2 $(M_2)$ given our data. This is just the circumstance where we've got two models that we want to compare. How could we calculate either of these terms? It is not difficult if we realize that each of these terms is essentially the posterior distribution. 

The model comparison problem we mentioned here could be formulated as
$$P(M_1|{\bf y}) \text{ vs } P(M_2|{\bf y})$$
We denote the posterior probability of model $M_1$ as $P(M_1|{\bf y})$, where ${\bf y}$ is a given collection of data.

By the classic Bayes rule, we have that
\begin{equation}
\label{eq:bayes1}
P(M_1|{\bf y})=\frac{P({\bf y}|M_1)P(M_1)}{P({\bf y})}\end{equation}
where we call $P({\bf y}|M_1)$ the marginal likelihood of model $M_1$, it describes the probability of getting data ${\bf y}$ under the assumptions of model $M_1$. $P(M_1)$ is the prior probability of $M_1$. In circumstance where we have two models that we are choosing between, $P(M_1)$ is usually just $1-P(M_2)$. Finally, the denominator $P({\bf y})$ is also a marginal likelihood, except now it is the marginal likelihood over both models.

To calculate each of these terms in detail, let's start with the marginal likelihood of model $M_1$, $P({\bf y}|M_1)$. Here we use another application of Bayes rule. In traditional Bayesian inference, we are trying to find out the posterior probability, or the probability density of some parameter vector $\boldsymbol{\theta}_i$ in the model $M_i$ conditional on data {\bf y}.
\begin{equation}
\label{eq:bayes2}
P(\boldsymbol{\theta}_i|{\bf y})=\frac{P({\bf y}|\boldsymbol{\theta}_i)P(\boldsymbol{\theta}_i)}{P({\bf y})}\end{equation}
By writing down our Bayes rule for inference implicitly, we are typically conditioning on a single model. So we could write as
\begin{equation}
\label{eq:bayes3}
P(\boldsymbol{\theta}_1|{\bf y},M_1)=\frac{P({\bf y}|\boldsymbol{\theta}_1,M_1)P(\boldsymbol{\theta}_1|M_1)}{P({\bf y}|M_1)}\end{equation}
Then we can see that the denominator in \ref{eq:bayes3}, $P({\bf y}|M_1)$, is that we are trying to calculate. Then we could calculate $P({\bf y}|M_1)$ by integrating out a continuous parameter vector, or summing for a discrete one, about the numerator of \ref{eq:bayes3} as
\begin{equation}
\label{eq:1}P({\bf y}|M_1) = \int P({\bf y}|\boldsymbol{\theta}_1,M_1)P(\boldsymbol{\theta}_1|M_1) d \boldsymbol{\theta}_1\end{equation}
And essentially here by integrating with respect to $\boldsymbol{\theta}_1$, what we are doing is marginalizing out our $\boldsymbol{\theta}_1$ dependence from our numerator, which is why we get a marginal density on bottom of \ref{eq:bayes3} which doesn't depend on $\boldsymbol{\theta}_1$. Note that, $\boldsymbol{\theta}_1$ is a parameter vector, in other words, we've got a model with a lot of parameters. Then \ref{eq:1} will be a high dimensional integral. That is some of the difficulties behind calculating the marginal likelihoods.

About how to calculate the denominator of \ref{eq:bayes1}, the idea is that it is the marginal likelihood of our data across both models, i.e. marginalized over the model choice. So
\begin{equation}
\label{eq:2}
P({\bf y})=P({\bf y}|M_1)P(M_1)+P({\bf y}|M_2)P(M_2)
\end{equation}
Importantly, this denominator term, $P({\bf y})$, is the same whether we're working out $P(M_1|{\bf y})$ or $P(M_2|{\bf y})$, because it contains contributions from each of the models. So the odds for $M_1$ verses $M_2$ given data then is 
$$\frac{P(M_1|{\bf y})}{P(M_2|{\bf y})}=\frac{P({\bf y}|M_1)}{P({\bf y}|M_2)}\times \frac{P(M_1)}{P(M_2)}$$
The first term on the right is called the Bayes factor (BF). So the Bayes factor is defined as the ratio of the marginal likelihoods for the two models that we are comparing:
$$\text{BF}=\frac{P({\bf y}|M_1)}{P({\bf y}|M_2)}$$
\begin{table}
\begin{center}
\caption{Harold Jeffreys' scale for interpretation of Bayes factor}
\label{tbl:jeff}
\setlength{\tabcolsep}{0.5 em} 
{\renewcommand{\arraystretch}{1.5}
   \begin{tabular}{|c|c|} \hline
   BF & Strength of evidence for $M_1$ \cr \hline  \hline
   $< 10^0$                    & Negative                  \cr \hline
   $10^0$     to $10^{1/2}$    & Barely worth mentioning   \cr \hline
   $10^{1/2}$ to $10^1$        & Substantial               \cr \hline
   $10^1$     to $10^{3/2}$    & Strong                    \cr \hline
   $10^{3/2}$ to $10^2$        & Very strong               \cr \hline
   $> 10^2$                    & Decisive                  \cr \hline
   \end{tabular}
}
\end{center}
\end{table}
Harold Jeffreys gave a widely used scale for interpretation of BF \cite{jeff} as shown in table \ref{tbl:jeff}.

There are some issues with using this kind of framework, particularly using the odds $P(M_1|{\bf y})/P(M_2|{\bf y})$, to compare between models:
\begin{itemize}
    \item Difficulties in calculating the marginal likelihoods $P({\bf y}|M_i)$ not only because they are inherently a high dimensional integration or sum but also how the two terms that we are integrating or summing over interact with one another, which makes this integration pathological to calculate.
    
    \item Marginal likelihoods $P({\bf y}|M_i)$ are very sensitive to our choice of the prior for each of the models in Bayes rule for inference, i.e. to $P(\boldsymbol{\theta}_i|M_i)$. Marginal likelihoods could change significantly as we change our prior on parameters $\boldsymbol{\theta}_i$, even if those changes to the prior on parameters do not affect the posterior $P(M_i|{\bf y})$ particularly much. This kind of sensitivity is not preferable for a model comparison framework, since we don't want to change the prior $P(\boldsymbol{\theta}_i|M_i)$ if it does not affect our posterior $P(M_i|{\bf y})$. 
    
    \item In practice it is very hard to come up with sensible ways to ascribing prior probabilities $P(M_i)$. Particularly when you consider that comparing $M_1$, which is relatively a simple model with few parameters, with $M_2$, which is a relatively complex model. Surely in this case, we might want to assign less probability to $M_2$ than $M_1$. But exactly how much less probability we should give it? 
    \item Even if we do what a lot of people do, they just simply set the ratio $P(M_1)/P(M_2)=1$, there are still issues using the Bayes Factor to do the model comparison. Like, say what is the cutoff I prefer $M_1$ over $M_2$? Is BF$=1.00001$ makes a difference to BF$=1$?
    
\end{itemize}

In \cite{andrew1}, Andrew Gelman says that the correct way or a better way to do model comparison is via measures of predictive accuracy. These are things like Widely Applicable Information Criterion (WAIC) \cite{waic} or leave-one-out cross-validation (LOO-CV). This idea provides a much more nuanced way of doing model comparison because you could select your cross-validation data set to echo the eventual use of your model as opposed to Bayes factor framework which is very rigid in the way you do model comparison. 

\subsubsection{Approximation methods for Bayes factor}
Due to the computational difficulties and other issues we mentioned in the previous section, usually exact Bayes factor is not used directly by scientists in their research. Alternatively, it is always approximated by Bayesian information criterion (BIC), or Schwarz information criterion (also SIC, SBC, SBIC) as brought out by Schwarz in 1978 \cite{chiwazi}, or a variant BIC \cite{ray}. In this section, we will give a derivation of BIC using Laplace method for approximation integrals.

\begin{definition}
\label{bicdefinition}
The Bayesian information criterion (BIC) for candidate model $M$ is defined as
\begin{equation}
    \text{BIC}=-2\log \L(\hat{\boldsymbol{\theta}}|{\bf y})+p\log n
\end{equation}
where $\hat{\boldsymbol{\theta}}$ is the  maximum likelihood estimate (MLE) of $\boldsymbol{\theta}$ that maximize the likelihood function $\L(\boldsymbol{\theta}|{\bf y})$, $p$ is the number of parameters in the model, i.e. the dimension of $\boldsymbol{\theta}$, $|\boldsymbol{\theta}|$, and $n$ is the number of observations, i.e. $|{\bf y}|$.
\end{definition}
In practice, BIC is computed for each of the candidate models, and the model with the smallest BIC is selected as the best model. However, Schwarz's BIC was only justified under the assumption of independent, identically distributed (i.i.d.) observations, and only under linear models with the likelihood is from the regular exponential family \cite{chiwazi}. These limitations are the motivation of our research. We generalized the original BIC to the mixed-effects models in which the observations are correlated within the subjects and to other more general models. We will talk about this in detail in next section.

\begin{theorem}
\label{bic}
The log marginal likelihood $P({\bf y}|M)$ for model $M$ could be approximated as
\begin{equation}
    \log P({\bf y}|M) \approx \log \L(\hat{\boldsymbol{\theta}}|{\bf y})-\frac{p}{2}\log n
\end{equation}
where $\hat{\boldsymbol{\theta}}$ is the  MLE of $\boldsymbol{\theta}$ that maximize the likelihood function $\L(\boldsymbol{\theta}|{\bf y})$.
\end{theorem}

The basic idea in the proof is the Laplace's method for approximating an integral. Suppose the function $f(x)$ is a twice continuously differentiable function on $[a,b]$ with a unique global maximum at $x_0 \in (a,b)$, additionally $f''(x_0)<0$. Then
$$\int^b_a e^{\lambda f(x)} \approx \sqrt{\frac{2\pi}{\lambda |f''(x_0)|}}e^{\lambda f(x_0)} \text{ as }\lambda \to \infty.$$
The detailed proof of Theorem \ref{bic} is as follows.
\begin{proof}
From \ref{eq:1} we know that the marginal likelihood of model $M$ could be written as
\begin{equation}
\label{pyM}
\begin{aligned}
    P({\bf y}|M) & =\int f({\bf y}|\boldsymbol{\theta}, M)g(\boldsymbol{\theta}|M)d \boldsymbol{\theta} \xlongequal{\text{short as}} \int f({\bf y}|\boldsymbol{\theta})g(\boldsymbol{\theta})d \boldsymbol{\theta}\\
    & = \int \exp\left\{\log \left[f({\bf y}|\boldsymbol{\theta})g(\boldsymbol{\theta})\right]\right\}d \boldsymbol{\theta}
\end{aligned}
\end{equation}
where $f({\bf y}|\boldsymbol{\theta})$ is the density of the data given the parameters $\boldsymbol{\theta}$ and $g(\boldsymbol{\theta})$ is the prior density of the parameters $\boldsymbol{\theta}$ under model $M$.

Define $\tilde{\boldsymbol{\theta}}$ as the mode of the posterior distribution $h(\boldsymbol{\theta}|{\bf y})$, at where $f({\bf y}|\boldsymbol{\theta})g(\boldsymbol{\theta})$ attains its maximum then $\log \left[f({\bf y}|\boldsymbol{\theta})g(\boldsymbol{\theta})\right]$ attains its maximum also. 
$$\tilde{\boldsymbol{\theta}} = \argmax_{\boldsymbol{\theta}} h(\boldsymbol{\theta}|{\bf y})=\argmax_{\boldsymbol{\theta}} \frac{f({\bf y}|\boldsymbol{\theta})g(\boldsymbol{\theta})}{\int f({\bf y}|\boldsymbol{\theta})g(\boldsymbol{\theta})d \boldsymbol{\theta}}$$
\begin{remark}
When we use the noninformative prior, for example let $g(\boldsymbol{\theta})=1$, then $\L(\boldsymbol{\theta}|{\bf y})=f({\bf y}|\boldsymbol{\theta})$ attains its maximum at the MLE of $\boldsymbol{\theta}$, i.e.  $\tilde{\boldsymbol{\theta}}=\hat{\boldsymbol{\theta}}$. To simplify the notation, we will use $\hat{\boldsymbol{\theta}}$ as our posterior mode.
\end{remark}
We can now expand $Q=\log \left[f({\bf y}|\boldsymbol{\theta})g(\boldsymbol{\theta})\right]$ about $\hat{\boldsymbol{\theta}}$ using Taylor's theorem and omit the remainder term as:
\begin{equation}
    \log \left[f({\bf y}|\boldsymbol{\theta})g(\boldsymbol{\theta})\right] \approx \log \left[f({\bf y}|\hat{\boldsymbol{\theta}})g(\hat{\boldsymbol{\theta}})\right] + (\boldsymbol{\theta}-\hat{\boldsymbol{\theta}}) \nabla_{\boldsymbol{\theta}}Q|_{\hat{\boldsymbol{\theta}}}+\frac{1}{2}(\boldsymbol{\theta}-\hat{\boldsymbol{\theta}})^TH_{\hat{\boldsymbol{\theta}}}(\boldsymbol{\theta}-\hat{\boldsymbol{\theta}})
\end{equation}
where $\nabla_{\boldsymbol{\theta}}Q$ is the gradient of $Q$ such that
$$\left(\nabla_{\boldsymbol{\theta}}Q\right)_i=\frac{\partial}{\partial \theta_i}Q$$
and $H_{\hat{\boldsymbol{\theta}}}$ is the Hessian matrix of dimension $|\boldsymbol{\theta}|\times |\boldsymbol{\theta}|=p\times p$ such that $$H_{ml}=\left.\frac{\partial^2}{\partial \theta_m \partial \theta_l}Q \right|_{\hat{\boldsymbol{\theta}}}$$
Since $Q$ attains its maximum at $\hat{\boldsymbol{\theta}}$, $\nabla_{\boldsymbol{\theta}}Q|_{\hat{\boldsymbol{\theta}}}=0$. Also since $Q$ is concave around $\hat{\boldsymbol{\theta}}$ and the Hessian matrix $H_{\hat{\boldsymbol{\theta}}}$ is negative definite. Denote $\tilde{H}_{\hat{\boldsymbol{\theta}}}=-H_{\hat{\boldsymbol{\theta}}}$, so $\tilde{H}_{\hat{\boldsymbol{\theta}}}$ is positive definite and symmetric. So equation \ref{pyM} could be approximated as
\begin{equation}
\label{pyM2}
\begin{aligned}
p({\bf y}|M) & \approx \int \exp \left\{ Q\left.\right|_{\hat{\boldsymbol{\theta}}} + (\boldsymbol{\theta}-\hat{\boldsymbol{\theta}}) \nabla_{\boldsymbol{\theta}}Q|_{\hat{\boldsymbol{\theta}}}-\frac{1}{2}(\boldsymbol{\theta}-\hat{\boldsymbol{\theta}})^T\tilde{H}_{\hat{\boldsymbol{\theta}}}(\boldsymbol{\theta}-\hat{\boldsymbol{\theta}}) \right\} d \boldsymbol{\theta}\\
& = \exp\left(Q\left.\right|_{\hat{\boldsymbol{\theta}}}\right) \int \exp \left\{ -\frac{1}{2}(\boldsymbol{\theta}-\hat{\boldsymbol{\theta}})^T\tilde{H}_{\hat{\boldsymbol{\theta}}}(\boldsymbol{\theta}-\hat{\boldsymbol{\theta}}) \right\} d \boldsymbol{\theta}\\
& \xlongequal{\text{let } {\bf X}= \boldsymbol{\theta}-\hat{\boldsymbol{\theta}}} \exp\left(Q\left.\right|_{\hat{\boldsymbol{\theta}}}\right) \int \exp \left\{ -\frac{1}{2} {\bf X}^T \tilde{H}_{\hat{\boldsymbol{\theta}}} {\bf X} \right \} d{\bf X}
\end{aligned}
\end{equation}
Since $\tilde{H}_{\hat{\boldsymbol{\theta}}}$ is symmetric, we could do the spectral decomposition for $\tilde{H}_{\hat{\boldsymbol{\theta}}}$ as
$$\tilde{H}_{\hat{\boldsymbol{\theta}}}=S^T \Lambda S$$
where $\Lambda$ is a diagonal matrix whose diagonal elements are eigenvalues of $\tilde{H}_{\hat{\boldsymbol{\theta}}}$, and the columns of $S$ are the corresponding independent eigenvectors. Note that $S$ is full rank and orthogonal, i.e. $S^TS=SS^T=I$. Let's change of variable as ${\bf X}=S^T{\bf U}$, or write it in detail as
$${\bf X}=\begin{pmatrix} X_1 \\ X_2 \\ \vdots \\ X_p \end{pmatrix}=S^T{\bf U}=\begin{pmatrix}s_{11} & s_{21} & \cdots & s_{p1} \\ s_{12} & s_{22} & \cdots & s_{p2} \\ \vdots & \vdots & \cdots & \vdots \\ s_{1p} & s_{2p} & \cdots & s_{pp}\end{pmatrix}\begin{pmatrix} U_1 \\ U_2 \\ \vdots \\ U_p \end{pmatrix}=\begin{pmatrix} \sum^p_{i=1}s_{i1}U_i \\ \sum^p_{i=1}s_{i2}U_i \\ \vdots \\ \sum^p_{i=1}s_{ip}U_i\end{pmatrix}$$
In other words $X_m=\sum^p_{i=1}s_{im}U_i$, which implies
$$\frac{\partial}{\partial U_l}X_m=s_{lm}$$
So the Jacobian matrix ${\bf J}$ would be 
$${\bf J}=\begin{pmatrix}\frac{\partial X_1}{\partial U_1}  & \cdots & \frac{\partial X_1}{\partial U_p}  \\ \vdots & \cdots & \vdots \\  \frac{\partial X_p}{\partial U_1}   & \cdots & \frac{\partial X_p}{\partial U_p} \end{pmatrix}=\begin{pmatrix}s_{11} & s_{21} & \cdots & s_{p1} \\ s_{12} & s_{22} & \cdots & s_{p2} \\ \vdots & \vdots & \cdots & \vdots \\ s_{1p} & s_{2p} & \cdots & s_{pp}\end{pmatrix}=S^T$$
Also note that $S^T$ is orthogonal, so
$$\det(S^TS)=\det(S^T)\det(S)=\left[\det(S^T)\right]^2=\det(I)=1$$
So $$\det({\bf J})=\det(S^T)= \pm 1$$
Let's continue with equation \ref{pyM2}, so we have
\begin{equation}
\label{pyM3}
\begin{aligned}
p({\bf y}|M) & \approx  \exp\left(Q\left.\right|_{\hat{\boldsymbol{\theta}}}\right) \int \exp \left\{ -\frac{1}{2} {\bf U}^TS S^T\Lambda S S^T{\bf U}  \right \} \left|\det({\bf J})\right|d{\bf U}\\
& = \exp\left(Q\left.\right|_{\hat{\boldsymbol{\theta}}}\right) \int \exp \left\{ -\frac{1}{2} {\bf U}^T \Lambda {\bf U}  \right \} d{\bf U}\\
& = \exp\left(Q\left.\right|_{\hat{\boldsymbol{\theta}}}\right) \int \exp \left\{ -\frac{1}{2} \sum^p_{i=1}U_i^2\lambda_i  \right \} d{\bf U}
\end{aligned}
\end{equation}
the last step here could be more clear if we write out the details as  
$${\bf U}^T \Lambda {\bf U}=\begin{pmatrix}U_1 & U_2 &  \cdots & U_p \end{pmatrix}\begin{pmatrix}\lambda_1 & \ & \ & \ & \\ \ & \lambda_2 & \ & \  \\ \ & \ & \ddots & \ \\ \ & \ & \ & \lambda_p  \end{pmatrix} \begin{pmatrix}U_1 \\ U_2 \\  \vdots \\ U_p \end{pmatrix}=\sum^p_{i=1}U_i^2\lambda_i$$
Note here we have a $p-$dimensional integration, each of them is a one dimensional integration of a normal kernel and could be evaluated using the property of the normal density. 
\begin{equation}
\label{pyM4}
\begin{aligned}
p({\bf y}|M) &  \approx \exp\left(Q\left.\right|_{\hat{\boldsymbol{\theta}}}\right) \int \cdots \int \exp \left\{ -\frac{1}{2} \sum^p_{i=1}U_i^2\lambda_i  \right \} dU_1 \cdots dU_p \\ 
& = \exp\left(Q\left.\right|_{\hat{\boldsymbol{\theta}}}\right)  \prod^p_{i=1} \int \exp \left\{ -\frac{1}{2} U_i^2\lambda_i  \right \} dU_i \\
& = \exp\left(Q\left.\right|_{\hat{\boldsymbol{\theta}}}\right)  \prod^p_{i=1} \sqrt{\frac{2\pi}{\lambda_i}}\\
& = \exp\left(Q\left.\right|_{\hat{\boldsymbol{\theta}}}\right) \frac{(2\pi)^{p/2}}{\prod^p_{i=1} \lambda_i^{1/2}}\\
& =  \exp\left(Q\left.\right|_{\hat{\boldsymbol{\theta}}}\right) \frac{(2\pi)^{p/2}}{\left[\det \left(\tilde{H}_{\hat{\boldsymbol{\theta}}}\right)\right]^{1/2}} 
\end{aligned}
\end{equation}
where the last step is using the fact that 
$$\det\left(\tilde{H}_{\hat{\boldsymbol{\theta}}} \right)=\det\left(S^T\Lambda S\right)=\det(S^T)\det(\Lambda)\det(S)=\det(\Lambda)=\prod^p_{i=1}\lambda_i$$
Thus, the log marginal likelihood of model $M$ has the relation
\begin{equation}
\label{logpyM}
\log P({\bf y}|M) \approx  \log f({\bf y}|\hat{\boldsymbol{\theta}}) + \log g(\hat{\boldsymbol{\theta}}) +\frac{p}{2}\log(2\pi)-\frac{1}{2}\log \left[\det \left( \tilde{H}_{\hat{\boldsymbol{\theta}}} \right)\right]
\end{equation}

\begin{remark}
We will derive the results further under the noninformative priors, i.e. when $g(\boldsymbol{\theta})=1$. Also, we assume the observations $y_1, y_2, \cdots, y_n$ are independent and identically distributed (i.i.d.), in which case the Hessian matrix becomes the observed Fisher information matrix. Also, assume $n$ is large which coincides with the condition $ \lambda \to \infty$ in Laplace's method and more importantly allows us to use the weak law of large numbers. I may want the readers to keep reminding themselves of this remark in the rest sections or even throughout this project since this is the core motivation of our research. 
\end{remark}
$$\begin{aligned}
\tilde{H}_{ml} & =-\left.\frac{\partial^2 \log \left[f({\bf y}|\boldsymbol{\theta}) g(\boldsymbol{\theta}) \right]}{\partial \theta_m \theta_l}\right|_{\boldsymbol{\theta}
=\hat{\boldsymbol{\theta}}}\\
& = -\left.\frac{\partial^2 \log \left[f({\bf y}|\boldsymbol{\theta}) \right]}{\partial \theta_m \theta_l}\right|_{\boldsymbol{\theta}  =\hat{\boldsymbol{\theta}} } \\
& =  -\left.\frac{\partial^2 \log \left[\prod^n_{i=1}f(y_i|\boldsymbol{\theta}) \right]}{\partial \theta_m \theta_l}\right|_{\boldsymbol{\theta}  =\hat{\boldsymbol{\theta}} }\\
& = -\left.\frac{\partial^2  \sum^n_{i=1} \log \left[f(y_i|\boldsymbol{\theta}) \right]}{\partial \theta_m \theta_l}\right|_{\boldsymbol{\theta}  =\hat{\boldsymbol{\theta}} } \\
& =-\left.\frac{\partial^2 \sum^n_{i=1}  \log \L(\boldsymbol{\theta}|y_i)}{\partial \theta_m\theta_l}\right|_{\boldsymbol{\theta}=\hat{\boldsymbol{\theta}}}\\
& = -\left.\frac{\partial^2 \left[ \frac{1}{n}\sum^n_{i=1} n \log \L(\boldsymbol{\theta}|y_i)\right]}{\partial \theta_m\theta_l}\right|_{\boldsymbol{\theta}=\hat{\boldsymbol{\theta}}}
\end{aligned}$$
Consider $n \log \L(\boldsymbol{\theta}|y_i)$ as a new random variable, by weak law of large number we obtain
$$\frac{1}{n}\sum^n_{i=1} n \log \L(\boldsymbol{\theta}|y_i) \xrightarrow{\text{in probability}}{\mathbb E} \left[n \log \L(\boldsymbol{\theta}|y_i) \right] \text{ for }\forall i \text{ as } n \to \infty.$$
So each element in the observed Fisher information matrix is 
$$\begin{aligned}
\tilde{H}_{ml} & = -\left.\frac{\partial^2 {\mathbb E} \left[n \log \L(\boldsymbol{\theta}|y_1) \right]}{\partial \theta_m\theta_l}\right|_{\boldsymbol{\theta}=\hat{\boldsymbol{\theta}}}\\ 
& = n\left\{- \left.\frac{\partial^2 {\mathbb E} \left[ \log \L(\boldsymbol{\theta}|y_1) \right]}{\partial \theta_m\theta_l}\right|_{\boldsymbol{\theta}=\hat{\boldsymbol{\theta}}}\right \}\\ 
& = n[I(\boldsymbol{\theta})]_{ml}
\end{aligned}$$
where $I(\boldsymbol{\theta})$ is the Fisher information matrix for a single data $y_1$. So
$$\det\left(\tilde{H}_{\hat{\boldsymbol{\theta}}} \right)=n^p \det \left(I(\boldsymbol{\theta}) \right)$$
Plugging this back to equation \ref{logpyM} and as $n \to \infty$ we only keep the terms involving sample size $n$, we have
\begin{equation}
\label{logpyM2}
\begin{aligned}
\log P({\bf y}|M) & \approx  \log \L(\hat{\boldsymbol{\theta}}|{\bf y}) + \log g(\hat{\boldsymbol{\theta}}) +\frac{p}{2}\log(2\pi) -\frac{1}{2}\log \left[\det \left( \tilde{H}_{\hat{\boldsymbol{\theta}}} \right)\right]\\
& = \log \L(\hat{\boldsymbol{\theta}}|{\bf y}) -\frac{p}{2}\log n-\frac{1}{2}\log \left[\det \left( I(\boldsymbol{\theta}) \right)\right]\\
& =  \log \L(\hat{\boldsymbol{\theta}}|{\bf y}) -\frac{p}{2}\log n
\end{aligned}
\end{equation}
\end{proof}
A lot of literature just keep the result in equation \ref{logpyM2} as the definition of BIC for model $M$, but I will use the definition for BIC as shown earlier in Definition \ref{bicdefinition} which is a variation of equation \ref{logpyM2} since it coincides with the formula for BIC in the programming language {\tt R}, which will easier to interpret in our later numerical experiments section. 

To remind ourselves, the Definition \ref{bicdefinition} for BIC is 
$$\text{BIC}=-2\log \L(\hat{\boldsymbol{\theta}}|{\bf y})+p\log n=-2\log P({\bf y}|M)$$
So when given two models, say $M_1$ and $M_2$, we will calculate the BIC for both of them and the Bayes factor for model comparison between $M_1$ and $M_2$ could then be approximated as 
$$\begin{aligned}\text{BF} & =\frac{P({\bf y}|M_1)}{P({\bf y}|M_2)}\\
& =\exp\left\{\log \left[\frac{P({\bf y}|M_1)}{P({\bf y}|M_2)} \right]\right\}\\
& =\exp\left\{\log P({\bf y}|M_1) - \log P({\bf y}|M_2) \right\} \\
& \approx \exp \left\{-\frac{1}{2}(\text{BIC}_1-\text{BIC}_2)\right\}\\
& =\exp \left\{-\frac{1}{2}\Delta \text{BIC}\right\}
\end{aligned}$$
From this, we could see that it is the difference between two BICs that matters, the model with the lowest BIC is always considered to be the best. The strength of the evidence against the model with the higher BIC value can be summarized \cite{bictable} as in Table \ref{bictable}. Readers could compare this table with the previous Table \ref{tbl:jeff} which is the scale of BF in the model comparison.
These two tables are corresponding with each other approximately. 
\begin{table}
\begin{center}
\caption{Strength of Evidence Provided by the Difference in BIC Values.}
\label{bictable}
\setlength{\tabcolsep}{0.5 em} 
{\renewcommand{\arraystretch}{1.5}
   \begin{tabular}{|c|c|} \hline
   $\Delta$BIC & Evidence against higher BIC \cr \hline  \hline
   $0$   to  $2$    & Not worth more than a bare mention   \cr \hline
   $2$   to  $6$    & Positive              \cr \hline
   $6$   to  $10$   & Strong                \cr \hline
   $> 10$           & Very Strong           \cr \hline
   \end{tabular}
}
\end{center}
\end{table}

\subsection{Linear Mixed-effects Models}
Linear mixed-effects models are an extension of simple linear models which include both fixed and random effects. Consider an example where we have $N$ patients, and we measure the blood pressure, age, weight, height, etc. at each morning during a week for each patient. We want to predict the blood pressure using the rest of the variables. If we assume that all the patients have the same slope and intercept relating blood pressure to age, weight, and height, then we can fit a regular linear model with blood pressure as the response and the other variables as the predictors. 

A mixed-effects model has both random and fixed effects. It usually happens when we have a model with a categorical predictor and the observations are divided into groups according to the category values. In our example, the categorical predictor could be the patient ID. Then the random effects can account for individual differences when a week's observations within persons are more correlated than observations between persons. 

A general linear mixed model with Gaussian errors for subjects $i$ (or patient $i$ in the example above) is \cite{13,14,longitudinal}
\begin{equation}
\label{mixed}
    {\bf y}_i={\bf X}_i{\boldsymbol{\beta}}+{\bf Z}_i{\boldsymbol{\alpha}}_i+{\boldsymbol{\epsilon}}_i, \ \ i=1, 2, \cdots, N.
\end{equation}
where
\begin{itemize}
    \item ${\bf y}_i=\begin{pmatrix} y_{i1} \\ \vdots \\  y_{in_i}\end{pmatrix}$ is a column vector of length $n_i$ of the response variables for subject $i$, and $y_{ij}$ is the $j$th observation on the $i$th subject. For example, $y_{23}$ could be the blood pressure for the second patient measured on the third day of a week.
    
    \item ${\bf X}_i$ is an $n_i \times p$ matrix of observed variables, usually with the first column as all 1's, i.e. ${\bf X}_i=\begin{pmatrix}1 & x_{i;1,1} & x_{i;1,2} & \cdots & x_{i;1,(p-1)}\\
    1 & x_{i;2,1} & x_{i;2,2} & \cdots & x_{i;2,(p-1)}\\
    \vdots & \vdots & \vdots & \cdots & \vdots\\
    1 & x_{i;n_i, 1} & x_{i;n_i,2} & \cdots & x_{i;n_i,(p-1)}\\ 
    \end{pmatrix}$. According to the blood pressure example, the second column vector of ${\bf X}_i$ could be the weight of patient $i$ measured in a week. Similarly, the third column could be the height, and the fourth column could be the age, and so on.
    
    \item ${\boldsymbol{\beta}}$ is the unknown regression coefficients of length $p$, which is the fixed effects vector  need to be estimated.
    
    \item ${\bf Z}_i=\begin{pmatrix}1 & z_{i;1,1} & z_{i;1,2} & \cdots & z_{i;1,(q-1)}\\
    1 & z_{i;2,1} & z_{i;2,2} & \cdots & z_{i;2,(q-1)}\\
    \vdots & \vdots & \vdots & \cdots & \vdots\\
    1 & z_{i;n_i, 1} & z_{i;n_i,2} & \cdots & z_{i;n_i,(q-1)}\\ 
    \end{pmatrix}$ is an $n_i \times q$ matrix, usually with the first column all 1's, for the random effects, ${\boldsymbol{\alpha}}_i$.
    
    \item $\boldsymbol{\alpha}_i$ are the unknown random effects vectors of length $q$, which are assumed to be independently distributed across subjects with distribution $\boldsymbol{\alpha}_i \sim N(0, {\bf G})$.
    
    \item $\boldsymbol{\epsilon}_i$ is the random error vector which is assumed to be independent across subjects with distribution $\boldsymbol{\epsilon}_i \sim N(0, \boldsymbol{\Sigma}_i)$.
    
    \item Here we also assume that $\boldsymbol{\alpha}_i$ and $\boldsymbol{\epsilon}_i$ are independent.
\end{itemize}
Then the covariance matrix of the response ${\bf y}_i$ is
\begin{equation}
\label{varyi}
\begin{aligned}
\var({\bf y}_i) & =\var({\bf X}_i{\boldsymbol{\beta}}+{\bf Z}_i{\boldsymbol{\alpha}}_i+{\boldsymbol{\epsilon}}_i)\\
& = \var({\bf Z}_i{\boldsymbol{\alpha}}_i+{\boldsymbol{\epsilon}}_i)\\
& =\var({\bf Z}_i{\boldsymbol{\alpha}}_i) + \var({\boldsymbol{\epsilon}}_i)\\
& = {\bf Z}_i\var({\boldsymbol{\alpha}}_i) {\bf Z}_i^T + \var({\boldsymbol{\epsilon}}_i)\\
& = {\bf Z}_i {\bf G} {\bf Z}_i^T +\boldsymbol{\Sigma}_i
\end{aligned}
\end{equation}
We will see how this variance is important in the next section \ref{sec:BIC_ne}. Model \ref{mixed} is our main model for the whole research, in other words, we are considering improving the definition for BIC under the mixed-effects model case. Mixed-effects model plays an important role in model selection when the data is not independent, for example, clustering data or longitudinal data, etc. The improved BIC provides a more precise method to select between mixed-effects models which also gives a better theoretical foundation than the original BIC for the data that is not independent.


\section{Improved definitions for BIC}
The expression of BIC is under the assumption that the observations $y_1, y_2, \cdots, y_n$ are independent and identically distributed (i.i.d.), in which case the Hessian matrix becomes the observed Fisher information matrix. However, we do not always have i.i.d. samples in practice. For example, clustered observations tend to be more similar to each other within the same group than those observations in other groups, and longitudinal data is collected by measuring the same subject repeatedly\cite{longitudinal}.  In these scenarios, the assumption in BIC, the observations are independent, is not satisfied. The concept of effective sample size was brought up in many literatures like \cite{longitudinal, TESS, ess_blog}. I will give a detailed explanation about these ideas in the following sections, and improved BIC is defined by replacing the sample size $n$ in the original BIC expression with the effective sample size.



\subsection{New BIC using effective sample size}
\label{sec:BIC_ne}

The Bayesian Information Criterion (BIC) model selection procedure provides a consistent, compared with AIC, and easily performed method \cite{27}. However, the BIC expression differs from one software to another. Since in the penalty part, $p\log n$, the effective sample size, $n$, and the effective number of parameters, $p$, are not well defined in the non-iid observation circumstances such as in mixed-effects models. The $\log n$ penalty is implemented in the {\tt R} package {\tt nlme} \cite{23} and {\tt lme4} \cite{lp} and in the SPSS procedure MIXED \cite{28} where $n$ is the total number of observations, while the $\log N$ penalty is used in Monolix \cite{18}, saemix \cite{2} or in the SAS proc NLMIXED \cite{25} where $N$ is the number of subjects in mixed models. Two improved BICs are defined for general mixed-effects models using the effective sample size discussed in section \ref{sec:ess} and later in section \ref{sec:BICh}.

\subsubsection{The Effective Sample Size}
\label{sec:ess}
Let's start this section with an example mentioned in \cite{ess_blog}:
\begin{quotation}
    {\em On a scale of 0 to 10, how much does the average citizen of the Republic of Elbonia trust the president? You’re conducting a survey to find out, and you’re going to need a sample of 100 statistically independent individuals. Now you have to decide how to do this.

You could stand in the central square of the capital city and survey the next 100 people who walk by. But these opinions won’t be independent: probably politics in the capital isn’t representative of politics in Elbonia as a whole.
So you consider traveling to 100 different locations in the country and asking one Elbonian at each. But apart from anything else, this is far too expensive for you to do. Maybe a compromise would be OK. You could go to 10 locations and ask 20 people at each? 30? How many would you need to match the precision of 100 independent individuals - to have an “effective sample size” of 100?}
\end{quotation}

The precision mentioned above is typically defined as the reciprocal of the variance of an estimator. In practice, precision often refers to the closeness of two or more observations to each other. A high variance estimator has low precision and vice versa. Then in \cite{ess_blog}, we have a ``loose'' definition for the effective sample size as:

\begin{definition}
The {\bf Effective Sample Size} of an estimator is the number $n_{e}$ with the property that our estimator has the same precision (or variance) as the estimator got by sampling $n_{e}$ independent individuals.
\end{definition}

\begin{example}\cite{ess_blog}
When we have observations $y_1, y_2, \cdots, y_n$ are independent and identically distributed (iid). One estimator for the population mean $\mu$ could be the sample mean, i.e. 
$$\frac{1}{n}y_1+\frac{1}{n}y_2+\cdots + \frac{1}{n}y_n$$
Since $y_1, y_2, \cdots, y_n$ are iid, then the variance of this estimator is
$$\text{Var}\left(\frac{1}{n}y_1+\frac{1}{n}y_2+\cdots + \frac{1}{n}y_n\right)=n \cdot \text{Var}\left(\frac{1}{n}y_1\right)=n\cdot \frac{1}{n^2}\text{Var}(y_1)=\frac{\sigma^2}{n}$$
where $\sigma^2$ is the population variance. In other words, the precision of this estimator is
$$\text{precision}=\frac{n}{\sigma^2}$$
which increases as the sample size $n$ increases.

Now suppose we have a random sample $y_1, y_2, \cdots, y_n$ by which the observations do not need to be independent of each other. Let $\hat{\mu}$ be an estimator of the population mean $\mu$ with variance Var$\left(\hat{\mu}\right)$, then the precision of the estimator $\hat{\mu}$ is 
$$\text{precision}=\frac{1}{\text{Var}\left(\hat{\mu}\right)}$$
If we want to obtain the same precision by sampling $n_e$ observations independently, then 
$$\text{precision}=\frac{1}{\text{Var}\left(\hat{\mu}\right)}=\frac{n_e}{\sigma^2}$$
Hence, the {\bf effective sample size} of $\hat{\mu}$ is defined as
$$n_e=\frac{\sigma^2}{\text{Var}\left(\hat{\mu}\right)}=\text{population variance} \times \text{precision of the estimator}$$
\end{example}

\begin{definition}\cite{ess_blog}
The {\bf magnitude} $|R|$ of an invertible $n \times n$ matrix $R$ is the sum of all $n^2$ entries of $R^{-1}$.
\end{definition}
Since calculating the inverse is computationally expensive, however, solving linear systems is faster than computing inverses. Then to get $|R|$ we do not have to get the inverse of $R$. A much easier way would be using Gaussian elimination  to solve 
$$R{\bf w}= \mathbbm{1}$$
Then 
$$|R|=\mathbbm{1}^T R^{-1}\mathbbm{1}=\mathbbm{1}^T R^{-1}R{\bf w} =\mathbbm{1}^T {\bf w}= \sum w_i$$

\begin{theorem}
\label{ne}
The effective sample size of an unbiased linear estimator of the population mean is the magnitude of the sample correlation matrix $R$.
\end{theorem}

\begin{proof}
Suppose we have $n$ observations denoted as ${\bf y}^T=(y_1, y_2, \cdots, y_n)$ which are identically but not necessary independent distributed. Suppose we are only consider linear unbiased estimator of the population mean $\mu$, so the estimator could be written as 
$$\hat{\mu}={\bf a}^T{\bf y}$$
for some vector ${\bf a}^T=(a_1, \cdots, a_n)$, such that
$$\mathbbm{E}\left(\hat{\mu}\right)=\mu.$$
Then the variance of the estimator $\hat{\mu}$ is
$$\var\left(\hat{\mu}\right) ={\bf a}^T \var({\bf y}){\bf a} = \sigma^2 {\bf a}^T R{\bf a} $$
The effective sample size is
$$n_e=\frac{\sigma^2}{\var(\hat{\mu})}=\frac{1}{{\bf a}^T R{\bf a}}$$
So the maximum effective size among all possible linear unbiased estimator is defined as
$$\sup\left\{\frac{1}{{\bf a}^T R{\bf a}}: a 
\in \R^n, \sum^n_{i=1}a_i=1 \right\} $$
By Cauchy-Schwarz inequality, we have that the supremum is obtained at ${\bf a}={\bf w}/|R|$, i.e.
$$\text{ maximum }n_e=\frac{1}{\frac{{\bf w}^T}{|R|}R\frac{{\bf w}}{|R|}}=\frac{1}{\frac{{\bf w}^T}{|R|^2}\mathbbm{1}}=\frac{|R|^2}{\sum^n_{i=1}w_i}=\frac{|R|^2}{|R|}=|R|$$
OR usually approximately,
$$n_e=|R|=\text{magnitude of correlation matrix}$$
\end{proof}

The effective sample size does not need to be less than the total number of observations. One simple example would explain it.
\begin{example}\cite{ess_blog}
Suppose we have two observations $y_1$ and $y_2$. The correlation matrix is
$$R=\begin{pmatrix}1 & \rho \\ \rho & 1 \end{pmatrix}$$
Then $$R^{-1}=\frac{1}{1-\rho^2}\begin{pmatrix}1 & -\rho \\ -\rho & 1\end{pmatrix}$$
So the effective sample size equals the magnitude of the correlation matrix, which is
$$n_e=|R|=\frac{1}{1-\rho^2}\left(1-\rho+1-\rho\right)=\frac{2(1-\rho)}{1-\rho^2}=\frac{2}{1+\rho}$$
Then we could see that when $\rho$ is some negative number between $-1$ and 0, we will have an effective sample size which is greater than 2.
\end{example}

\subsubsection{$BIC_{n_e}$}
\begin{definition}

\label{bic_ne}
The Bayesian information criterion (BIC) using the effective sample size $n_e$ for candidate model $M$ is defined as
\begin{equation}
    \text{BIC}_{n_e}=-2\log \L(\hat{\boldsymbol{\theta}}|{\bf y})+p\log n_e
\end{equation}
where $\hat{\boldsymbol{\theta}}$ is the  maximum likelihood estimate (MLE) of $\boldsymbol{\theta}$ that maximize the likelihood function $\L(\boldsymbol{\theta}|{\bf y})$, $p$ is the number of parameters in the model, i.e. the dimension of $\boldsymbol{\theta}$, $|\boldsymbol{\theta}|$, and $n_e$ is the effective sample size defined in Theorem \ref{ne}, i.e. $n_e=|R|$=magnitude of correlation matrix $R$.

\end{definition}

We will bring in another innovative BIC in the following section, i.e. the $\text{BIC}_h$. Then we will conduct a simulation study to compare the performance of $\text{BIC}_{ne}$ and $BIC_{h}$ with two widely used BIC, the one with the sample size equals the number of subjects $N$ and the one with the sample size takes the total number of observations $n$.

\subsection{New BIC using hybrid sample size}
\label{sec:BICh}
According to \cite{donna}, 

{\em ... the information is of order $N$ (the number of units) for fixed effects with associated random effects, and of order $n$ (the number of total observations) for fixed effects with no associated random effect. This shows that the penalty term appearing in BIC will depend on which parameters are tested and on the specific variance-covariance structure of the model. 
}

This idea leads to the definition of a new BIC using the hybrid sample size.

\begin{theorem}

The Bayesian information criterion (BIC) using the hybrid sample size
for candidate model is defined as \cite{BICh}

\begin{equation}
    \text{BIC}_h=-2\log \L(\hat{\boldsymbol{\theta}}|{\bf y})+|\boldsymbol{\theta}_R|\log N +|\boldsymbol{\theta}_F|\log n
\end{equation}
where $\hat{\boldsymbol{\theta}}$ is the  maximum likelihood estimate (MLE) of $\boldsymbol{\theta}$ that maximize the likelihood function $\L(\boldsymbol{\theta}|{\bf y})$, $\boldsymbol{\theta}_R$ is the random components of parameter vector $\boldsymbol{\theta}$ in the model, and $|\boldsymbol{\theta}_R|$ the dimension of $\boldsymbol{\theta}_R$. Similarly,  $|\boldsymbol{\theta}_F|$ the dimension of $\boldsymbol{\theta}_F$, the fixed components of parameter vector $\boldsymbol{\theta}$. And same as before, $N$ is the number of subjects in the mixed effects model and $n$ is the number of total observations. 
\end{theorem}

Suppose we have a linear mixed-effects model 
$${\bf y}_i=X_i\psi_i+\epsilon_i$$
where $X_i$ is the design matrix and $\epsilon_i \sim N (0, \Sigma)$. A linear model for $\psi_i$ is also assumed as
\begin{equation}
\label{psi}
    \psi_i=C_i\beta+\eta_i
\end{equation}
where $\eta_i\sim N(0,\Omega)$. The vector of population parameters $\boldsymbol{\theta}$ includes $\beta$ and the parameter in $\Omega$.

Equation \ref{psi} considers models in which certain individual parameters in vector $\psi_i$ are random or fixed. Degenerated $\Omega$ could be a block-diagonal matrix as
$$\Omega=\begin{pmatrix}0 & 0 \\  0 & \Omega_R \end{pmatrix}$$
Here, when we use $\text{BIC}_h$, we assume that we know the structure of $\Omega$, i.e. we know which diagonal elements of $\Omega$ are zeros. From which we could see which parameters in the vector $\psi_i$ are fixed and which are random.

I will use an example in the simulation study section in \cite{BICh} to illustrate the $BIC_{h}$. The full proof of the definition of $\text{BIC}_h$ could be found in \cite{BICh} section 2.

\begin{example}\cite{BICh}
Suppose we have a linear mixed effects model as
$$y_{ij}=X_{ij}\psi_i+\epsilon_{ij}=\psi_{i0}+\psi_{i1}x_{ij}+\psi_{i2}x^2_{ij}+\epsilon_{ij}$$
for $i=1,...,N, j=1,...,n$, where $X_{ij}=\begin{pmatrix}1 & x_{ij} & x^2_{ij} \end{pmatrix}$ and $\psi_i=\begin{pmatrix}\psi_{i0} \\ \psi_{i1} \\ \psi_{i2} \end{pmatrix}$ with $\epsilon_{ij} \stackrel{iid}{\sim} N(0, \sigma^2)$.
Here we have
$$\psi_i=C_i\beta+\eta_i$$
with
$$C_i=\begin{pmatrix}1 & 0 & 0 & 0 & 0 \\ 0 & 1 & 0 & c_i & 0 \\ 0 & 0 & 1 & 0 & c_i  \end{pmatrix}; \hspace{1cm} \beta=\begin{pmatrix}\mu_0 \\ \mu_1 \\ \mu_2 \\ \alpha_1 \\ \alpha_2 \end{pmatrix}$$
and $\eta_i \stackrel{iid}{\sim} N(0,\Omega)$ with
$$\Omega=\begin{pmatrix}\omega_0^2 & 0 & 0 \\ 0& \omega^2_1 & 0 \\ 0 & 0 & \omega_2^2  \end{pmatrix}$$
The vector of population parameters $\boldsymbol{\theta}$ includes $\beta$ and the parameter in $\Omega$. Here, we reduce the model selection problem to select the non zero elements between $\alpha_1$ and $\alpha_2$ under four possible variance models, i.e. whether $\omega_1^2$ and $\omega^2_2$ are zeros or not. Thus, we have $4 \times 4=16$ possible situations as the combinations of the following:
$$\begin{aligned}M_1: & \alpha_1 = 0, \alpha_2 = 0 \\ 
M_2: & \alpha_1 \neq 0, \alpha_2 = 0\\
M_3: & \alpha_1 = 0, \alpha_2 \neq 0\\
M_4: & \alpha_1 \neq 0, \alpha_2 \neq 0\\
\end{aligned} \hspace{1cm} \text{ and } \hspace{1cm}
\begin{aligned} 
O_1: & \omega^2_1 = 0, \omega^2_2 = 0 \\ 
O_2: & \omega^2_1 \neq 0, \omega^2_2 = 0\\
O_3: & \omega^2_1 = 0, \omega^2_2 \neq 0\\
O_4: & \omega^2_1 \neq 0, \omega^2_2 \neq 0\\\end{aligned}$$

We will show one scenario in details here, and the full table of the elements of $\theta_R$, $\theta_F$ and penalization terms used by $\text{BIC}_N$, $\text{BIC}_n$ and $\text{BIC}_h$ will be shown in Table \ref{bichtable}.
Consider the case under $O_2$ and $M_1$,
$$O_2:  \omega^2_1 \neq 0, \omega^2_2 = 0; \hspace{1cm} M_1: \alpha_1 = 0, \alpha_2 = 0$$
Now
$$\Omega=\begin{pmatrix}\omega_0^2 & 0 & 0 \\ 0& \omega^2_1 & 0 \\ 0 & 0 & 0  \end{pmatrix} \text{ and } \beta=\begin{pmatrix}\mu_0 \\ \mu_1 \\ \mu_2 \\ 0 \\ 0 \end{pmatrix}$$
Then
$$\psi_i=C_i\beta+\eta_i=\begin{pmatrix}\mu_0\\ \mu_1+\alpha_1c_i \\ \mu_2+\alpha_2c_i \end{pmatrix}+\begin{pmatrix} \eta_{i0} \\ \eta_{i1} \\ \eta_{i2}\end{pmatrix}=\begin{pmatrix}\mu_0\\ \mu_1 \\ \mu_2 \end{pmatrix}+\begin{pmatrix} \eta_{i0} \\ \eta_{i1} \\ 0\end{pmatrix}=\begin{pmatrix}\mu_0+\eta_{i0}\\ \mu_1+\eta_{i1} \\ \mu_2 \end{pmatrix}$$
Hence we can see that in the vector $\psi_i$, $\mu_0$ and $\mu_1$ are random, $\mu_2$ is fixed. Don't forget the parameter $\boldsymbol{\theta}$ takes the elements in $\Omega$ and $\sigma^2$ as well.

\begin{table}
\begin{center}
\caption{Elements of $\theta_R$, $\theta_F$ and penalization terms used by different BICs}
\label{bichtable}
\setlength{\tabcolsep}{0.5 em} 
{\renewcommand{\arraystretch}{1.5}
   \begin{tabular}{|c|c|c|c|c|c|c|c|} \hline
   O & M  & $\beta_R$ & $\Omega_R$ & $\theta_F$ & BIC$_N$ & BIC$_n$ & BIC$_h$ \cr \hline  \hline
   \multirow{4}{*}{$O_1$} & $M_1$ & $\mu_0$ & $\omega_0$ & $\sigma^2, \mu_1, \mu_2$ & $5\log N$ & $5\log n$ & $2\log N + 3\log n$ \\  
   & $M_2$ & $\mu_0$ & $\omega_0$ & $\sigma^2, \mu_1, \mu_2, \alpha_1$ & $6\log N$ & $6\log n$ & $2\log N + 4\log n$ \\ & $M_3$ & $\mu_0$ & $\omega_0$ & $\sigma^2, \mu_1, \mu_2, \alpha_2$ & $6\log N$ & $6\log n$ & $2\log N + 4\log n$ \\
   & $M_4$ & $\mu_0$ & $\omega_0$ & $\sigma^2, \mu_1, \mu_2, \alpha_1, \alpha_2$ & $7\log N$ & $7\log n$ & $2\log N + 5\log n$ 
   \cr \hline

\multirow{4}{*}{$O_2$} & $M_1$ & $\mu_0, \mu_1$ & $\omega_0, \omega_1$ & $\sigma^2, \mu_2$ & $6\log N$ & $6\log n$ & $4\log N + 2\log n$ \\  
   & $M_2$ & $\mu_0, \mu_1, \alpha_1$ & $\omega_0, \omega_1$ & $\sigma^2, \mu_2$ & $7\log N$ & $7\log n$ & $5\log N + 2\log n$ \\ 
   & $M_3$ & $\mu_0, \mu_1$ & $\omega_0, \omega_1$ & $\sigma^2, \mu_2, \alpha_2$ & $7\log N$ & $7\log n$ & $4\log N + 3\log n$ \\
   & $M_4$ & $\mu_0, \mu_1, \alpha_1$ & $\omega_0, \omega_1$ & $\sigma^2, \mu_2,  \alpha_2$ & $8\log N$ & $8\log n$ & $5\log N + 3\log n$ 
   \cr \hline
   
   \multirow{4}{*}{$O_3$} & $M_1$ & $\mu_0, \mu_2$ & $\omega_0, \omega_2$ & $\sigma^2, \mu_1$ & $6\log N$ & $6\log n$ & $4\log N + 2\log n$ \\  
   & $M_2$ & $\mu_0, \mu_2$ & $\omega_0, \omega_2$ & $\sigma^2, \mu_1, \alpha_1$ & $7\log N$ & $7\log n$ & $4\log N + 3\log n$ \\ 
   & $M_3$ & $\mu_0, \mu_2, \alpha_2$ & $\omega_0, \omega_2$ & $\sigma^2, \mu_1$ & $7\log N$ & $7\log n$ & $5\log N + 2\log n$ \\
   & $M_4$ & $\mu_0, \mu_2, \alpha_2$ & $\omega_0, \omega_2$ & $\sigma^2, \mu_1,  \alpha_1$ & $8\log N$ & $8\log n$ & $5\log N + 3\log n$ 
   \cr \hline
   
   \multirow{4}{*}{$O_4$} & $M_1$ & $\mu_0, \mu_1, \mu_2$ & $\omega_0,\omega_1, \omega_2$ & $\sigma^2$ & $7\log N$ & $7\log n$ & $6\log N + \log n$ \\  
   & $M_2$ & $\mu_0, \mu_1, \mu_2, \alpha_1$ & $\omega_0,\omega_1, \omega_2$ & $\sigma^2$ & $8\log N$ & $8\log n$ & $7\log N + \log n$ \\ 
   & $M_3$ & $\mu_0, \mu_1, \mu_2, \alpha_2$ & $\omega_0,\omega_1, \omega_2$ & $\sigma^2$ & $8\log N$ & $8\log n$ & $7\log N + \log n$ \\
   & $M_4$ & $\mu_0, \mu_1, \mu_2, \alpha_1, \alpha_2$ & $\omega_0,\omega_1, \omega_2$ & $\sigma^2$ & $9\log N$ & $9\log n$ & $8\log N + \log n$ 
   \cr \hline

   \end{tabular}
}
\end{center}
\end{table}

\end{example}

\subsection{Simulation study}

\begin{figure}
\centering\includegraphics[width=13cm]{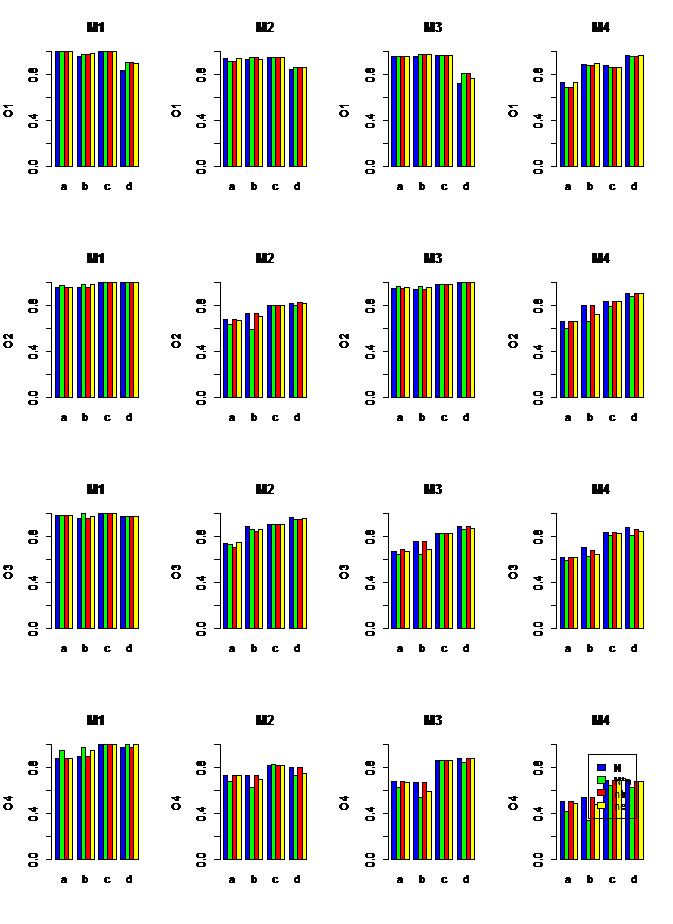}
\caption{Frequency of correct selection for the four BIC versions: BIC$_N$(blue), BIC$_n$(green), BIC$_{ne}$(yellow) and BIC$_h$(red) under different designs $a(N=20, n_{sub}=5), b(N=20, n_{sub}=100), c(N=100, n_{sub}=5), d(N=100, n_{sub}=100)$.}
\label{fig:numericaltable}
\end{figure}

We have four different designs as the number of subjects, $N=20$ or 100, and the number of observations per subject, $n_{sub}=5$ or 100, vary. For each of these 64 models, the involved parameters and variables are generated as follows: 
\begin{itemize}
    \item The $n$ observation points in the design matrix, $x_{i1}, ..., x_{in}$ were equally spaced in the interval $[0,10]$.
    \item The residual error was fixed as $\sigma^2=1$.
    \item $c_i \sim N(0,1)$.
    \item $\mu_0 \sim N(0.01, 1)$.
    \item $\mu_1 \sim N(0.005,1)$.
    \item $\mu_2 \sim N(0.0025, 1)$.
    \item $\alpha_1, \alpha_2 \sim N(0.01,1)$.
    \item $\omega^2_m \sim U[0.01,1.01], 0 \leq m \leq 2$.
\end{itemize}

Frequency of correct selection for the four BIC versions: BIC$_N$(blue), BIC$_n$(green), BIC$_{ne}$(yellow) and BIC$_h$(red) under different designs $a(N=20, n_{sub}=5), b(N=20, n_{sub}=100), c(N=100, n_{sub}=5), d(N=100, n_{sub}=100)$ are shown in Figure \ref{fig:numericaltable}. We could see that the two new BICs give an overall better selection procedure under different model selection problems.

\section{Conclusions and Future Work}
The deduction of the BIC formula in this chapter tells us that BIC is based on the assumption that the observations are independent,  identically  distributed (i.i.d.). When the real-world data does not satisfy this assumption, using BIC could be questionable since the sample size $n$ is not well defined. To apply BIC in such non-iid settings, like linear mixed model for clustered data, we define a new BIC, denoted as BIC$_{n_e}$ in our project, using the effective sample size $n_e$. The effective sample size of an estimator is defined as a function of the inverse of the information matrix, which would give the same precision as if we sample $n_{e}$ independent individuals.

Simulation study is conducted to compare the  performance of BIC$_{n_e}$ with two widely used BIC,  BIC$_N$, and BIC$_n$, and one innovation BIC defined in \cite{BICh} in which the penalty term is defined as a hybrid of the penalties in the classical BIC at two extreme cases. Using a simple linear
mixed-effects model, we have found that the performances of BIC$_N$ and BIC$_n$ differ a lot for different covariance structures. BIC$_{n_e}$ and BIC$_h$ behave as the best of the two standard BIC, whatever the random structure of the model. Moreover, our BIC$_{n_e}$ is easier to apply than BIC$_h$ since we do not require the structure of the covariance. Thus, BIC$_{n_e}$ has a more general assumption when we apply it to real-world data.

Our ongoing work focus on generalize our BIC$_{n_e}$ to more general cases besides linear mixed effect model. Non-linear mixed effect models and other more complex models are also widely used in practice. How to implement simulations on these data will be discussed in the future.


\end{document}